\documentclass[copyright]{eptcs}
\providecommand{\event}{J.~Dassow, G.~Pighizzini, B.~Truthe (Eds.): 11th International Workshop\\
on Descriptional Complexity of Formal Systems (DCFS 2009)} 
\providecommand{\volume}{3}
\providecommand{\anno}{2009}
\providecommand{\firstpage}{169} 
\providecommand{\eid}{16}
\usepackage{breakurl}        

\input{dcfs.tex}

\usepackage{graphicx} 
\usepackage{mathrsfs}
\usepackage{amssymb}
\usepackage{amsmath}
\usepackage{textcomp}

\newcommand{\lfam}{\mathscr{L}}
\newcommand{\border}{\texttt{\#}}
\newcommand{\cent}{\texttt{\textcent}}
\newcommand{\rtf}{\lfam_{rt}}
\newcommand{\func}[1]{\mathop{\mathrm{#1}}}
\newcommand{\comm}{\func{com}}
\newcommand{\mcomm}{\func{mcom}}
\newcommand{\scomm}{\func{scom}}
\newcommand{\ca}{\textrm{CA}}
\newcommand{\oca}{\textrm{OCA}}

\newcommand{\mcoca}[1]{\textrm{MC}(#1)\textrm{-OCA}}
\newcommand{\mcca}[1]{\textrm{MC}(#1)\textrm{-CA}}
\newcommand{\scoca}[1]{\textrm{SC}(#1)\textrm{-OCA}}
\newcommand{\scca}[1]{\textrm{SC}(#1)\textrm{-CA}}

\begin{document}

\title{Bounded Languages Meet Cellular Automata\\ with Sparse Communication}
\def\titlerunning{Bounded Languages Meet Cellular Automata with Sparse Communication}
\def\authorrunning{M.~Kutrib, A.~Malcher}
\author{Martin Kutrib \qquad Andreas Malcher
\institute{Institut f\"ur Informatik --
  Universit\"at Giessen\\
  Arndtstr.~2 -- 35392 Giessen -- Germany}
\email{\{kutrib,malcher\}@informatik.uni-giessen.de}
}

\maketitle

\begin{abstract}
Cellular automata are one-dimensional arrays of interconnected 
interacting finite automata. We investigate one of the
weakest classes, the real-time one-way cellular
automata, and impose an additional restriction on their
inter-cell communication by bounding the number 
of allowed uses of the links between cells. Moreover,
we consider the devices as acceptors for \emph{bounded}
languages in order to explore the borderline at which
non-trivial decidability problems of cellular automata classes
become decidable. It is shown that even devices with drastically
reduced communication, that is, each two neighboring cells may 
communicate only constantly often, accept bounded languages that
are not semilinear. If the number of communications is
at least logarithmic in the length of the input, several
problems are undecidable. The same result is obtained for
classes where the total number of communications during a computation
is linearly bounded.
\end{abstract}

\section{Introduction}

Cellular automata are linear arrays of identical copies of deterministic finite 
automata, where the single nodes, which are called cells, are homogeneously 
connected to both their immediate neighbors. They work synchronously at 
discrete time steps. In the general case, in every time step the state of 
each cell is communicated to its neighbors. That is, on the one hand the state 
is sent regardless of whether it is really required, and on the other hand, 
the number of bits sent is determined by the number of states.
Devices with bounded bandwidth of the inter-cell links are considered
in~\cite{Kutrib:2006:fcaricccc,
Umeo:2001:ltrcbi,
Umeo:2003:rtgpob,Worsch:2000ltlrcarc}.
In~\cite{Vollmar:1981:CAFN}
two-way cellular automata 
are considered where the number of proper state changes is bounded. 
There are strong relations to inter-cell communication. Roughly speaking, a cell
can remember the states received from its neighbors. As long as these do not change,
no communication is necessary.

Due to their temporal and structural restrictions 
real-time one-way cellular automata define one of the 
weakest classes of cellular automata. However, they are still
powerful enough to accept non-context-free (even non-semilinear) 
languages (see, e.\,g., the surveys~\cite{kutrib:2008:ca-cpv,kutrib:2009:calt}).
Moreover, almost all of the commonly investigated decidability questions
are known not to be semidecidable~\cite{Malcher:2002:dccadq}.
In order to explore the borderline at which non-trivial decidability 
problems become decidable, additional structural and computational
restrictions have been imposed.
Here, we investigate real-time one-way cellular automata where the 
communication is quantitatively measured by the number of uses of the links between
cells. Bounds on the sum of all communications of a computation, as well as
bounds on the maximal number of communications that may appear
between each two cells are considered. Reducing the communication
drastically, but still enough to have non-trivial devices, we
obtain systems where each two neighboring cells may 
communicate only constantly often, and systems where  
the total number of communications during a computation depends
linearly on the length of the input. However, it has been shown
in~\cite{Kutrib:2009:cssc:proc} that even these restrictions do not
lead to decidable properties.  

An approach often investigated and widely accepted is to consider a given
type of device for special purposes only, for example,
for the acceptance of languages having a certain structure or form.
{F}rom this point of view it is natural to start with unary 
languages (e.\,g.,~\cite{Book:1974:tlcc,
Chrobak:1986:FAUL,%
klein:2007:cdul,%
Mereghetti:2001:osbua,Pighizzini:2002:uloscjf}). For general real-time
one-way cellular automata it is known that they accept only regular unary
languages~\cite{Seidel:1979:LRSCA}. Since the proof is constructive, we
derive that the borderline in question has been crossed. So, we generalize
unary languages to bounded languages. For several devices it is known that
they accept non-semilinear languages in general, but only
semilinear bounded languages. Since for semilinear sets
several properties are decidable~\cite{Ginsburg:1966:MTCFL}, 
constructive proofs lead to decidable properties for these devices in 
connection with bounded languages 
\cite{csuhajvarju:1994:gs,ibarra:1970:sml,ibarra:1974:nslsbrmhpda,ibarra:1978:rbmcmdp}.  


\section{Definitions and preliminaries}\label{sect:def}

We denote the positive integers and zero $\{0,1,2,\dots\}$ by ${\mathbb{N}}$.
The empty word is denoted by~$\lambda$, the reversal of a word $w$ by $w^R$, 
and for the length of~$w$ we write~$|w|$. 
For the number of occurrences of a subword~$x$ in~$w$ we
use the notation~$|w|_x$.
We use~$\subseteq$ for inclusions and $\subset$ for strict inclusions.
A language $L$ over some alphabet $\{a_1,a_2,\dots,a_k\}$ is said to
be \emph{bounded}, if $L\subseteq  a^*_1a^*_2\cdots a^*_k$.

A cellular automaton is a linear array of
identical deterministic finite state machines, sometimes called cells.
Except for the leftmost cell and rightmost cell each one is connected 
to both its nearest neighbors.
We identify the cells by positive integers. The state transition depends on the
current state of each cell and on the information which is currently sent 
by its neighbors. The  information sent by a cell depends on its
current state and is determined by so-called communication functions.
The two outermost cells receive a boundary symbol on their free input lines
once during the first time step from the outside world. Subsequently, 
these input lines are never used again.
A formal definition is 

\begin{definition}
A \emph{cellular automaton} $(\ca)$ is a system
$\langle S,F,A,B,\border,b_l,b_r,\delta\rangle$, where 
$S$ is the finite, non\-emp\-ty set of \emph{cell states},
$F\subseteq S$ is the set of \emph{accepting states},
$A\subseteq S$ is the nonempty set of \emph{input symbols},
$B$ is the set of \emph{communication symbols},
$\border\notin B$ is the \emph{boundary symbol},
$b_l,b_r: S\to B\cup \{\bot\}$ are 
\emph{communication functions} which determine the information
\emph{to be sent} to the left and right neighbors, where 
$\bot$ means \emph{nothing to send}, and
$\delta:(B\cup\{\border,\bot\})\times S\times (B\cup\{\border,\bot\})
\to S$ is the \emph{local transition function}.  
\end{definition}

A \emph{configuration} of a cellular automaton 
$\langle S,F,A,B,\border,b_l,b_r,\delta\rangle$
at time $t\geq 0$ is a
description of its global state, which is actually a mapping
$c_t:[1,\dots,n] \to S$, for $n\geq 1$.
The operation starts at time 0 in a so-called \emph{initial configuration}. 
For a given input $w=a_1\cdots a_n\in A^+$ we set 
$c_{0,w}(i)=a_i$, for \hbox{$1\leq i\leq n$}. During the course of 
its computation a $\ca$ steps through a sequence of configurations, 
whereby successor configurations are computed
according to the global transition function~$\Delta$:
Let $c_t$, $t\geq 0$, be a configuration. Then its successor
configuration $c_{t+1}=\Delta(c_t)$ is as follows.
For $2\leq i\leq n-1$,
$$
c_{t+1}(i) =\delta(b_r(c_t(i-1)), c_t(i), b_l(c_t(i+1))),
$$
and for the leftmost and rightmost
cell we set  
\begin{align*}
c_{1}(1) &= \delta(\border, c_0(1),b_l(c_0(2))),\\
c_{t+1}(1) &= \delta(\bot, c_t(1),b_l(c_t(2))), \mbox{ for $t\geq 1$, and}\\
c_{1}(n) &= \delta(b_r(c_0(n-1)), c_0(n),\border),\\
c_{t+1}(n) &= \delta(b_r(c_t(n-1)), c_t(n),\bot), \mbox{ for $t\geq 1$}.
\end{align*}
Thus, the global transition function~$\Delta$ is induced by $\delta$.

\begin{figure}[b]
\centering
\includegraphics[scale=1]{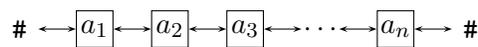}
\caption{A two-way cellular automaton.}
\end{figure}

An input $w$ is accepted by a $\ca$ $\mathcal{M}$ if at some time $i$
during the course of its computation the leftmost cell enters an accepting
state. The \emph{language accepted by $\mathcal{M}$} is denoted by
$L(\mathcal{M})$. Let $t:\mathbb{N}\to\mathbb{N}$, 
$t(n)\geq n$, be a mapping. 
If all $w\in L(\mathcal{M})$
are accepted with at most $t(|w|)$ time steps, then $\mathcal{M}$ is 
said to be of time complexity $t$. 

An important subclass of cellular automata are the so-called
\emph{one-way cellular automata} ($\oca$), where the flow of information
is restricted to one way from right to left. For a formal definition it
suffices to require that $b_r$ maps all states to $\bot$, and that
the leftmost cell does not receive the boundary symbol during the first time
step.

In the following we study the impact of communication in cellular automata.
Communication is measured by the number of uses of the links between cells.
It is understood that whenever a communication symbol not equal to 
$\bot$ is sent, a communication takes place. Here we do not distinguish
whether either or both neighboring cells use the link. More precisely,
the number of communications between cell $i$ and cell $i+1$
up to time step $t$ is defined by 
$$
\comm(i,t) = \left| \{\,j \mid 0\leq j < t \mbox{ and }
(b_r(c_j(i)) \ne \bot \mbox{ or } 
b_l(c_j(i+1)) \ne \bot)\,\}\right|.
$$
For computations we now distinguish the maximal number of communications
between two cells and the total number of communications.
Let $c_0,c_1,\dots, c_{t(|w|)}$ be the sequence of configurations
computed on input $w$ by some cellular automaton with time complexity 
$t(n)$, that is, the \emph{computation on $w$}. Then we define
\begin{eqnarray*}
\mcomm(w) &=& \max\{\, \comm(i,t(|w|)) \mid 1\leq i\leq |w|-1\,\}
\mbox{ and}\\
\scomm(w) &=& \sum_{i=1}^{|w|-1} \comm(i,t(|w|)).
\end{eqnarray*}
Let $f:\mathbb{N}\to\mathbb{N}$ be a mapping. 
If all $w\in L(\mathcal{M})$
are accepted with computations where $\mcomm(w) \leq f(|w|)$,
then $\mathcal{M}$ is said to be \emph{max communication bounded by~$f$.}
Similarly, if all $w\in L(\mathcal{M})$
are accepted with computations where \mbox{$\scomm(w) \leq f(|w|)$,}
then $\mathcal{M}$ is said to be \emph{sum communication bounded by $f$.}
In general, it is not expected to obtain tight bounds on the exact
number of communications but rather tight bounds up to a constant multiplicative
factor. For the sake of readability
we denote the class of $\ca$s that are max communication bounded by 
some function $g\in O(f)$ by $\mcca{f}$, where it is understood
that~$f$ gives the order of magnitude. Corresponding notation
is used for $\oca$s and sum communication
bounded $\ca$s and $\oca$s. ($\scca{f}$ and $\scoca{f}$).
The family of all languages which are accepted by some 
device $X$ with time complexity~$t$ is denoted by 
$\lfam_t(X)$.
In the sequel we are particularly interested in fast computations
and call the time complexity $t(n) = n$ \emph{real time}
and write $\rtf(X)$.

\section{Computational capacity}
\label{sec:comcapacity}

It has been shown in~\cite{Kutrib:2009:cssc:proc} 
that the family $\mcoca{1}$ contains 
the non-context-free languages
$$\{\,a_1^na_2^n\cdots a_k^n \mid n \geq 1\,\} \mbox{ and }
\{\,a^nb^mc^nd^m \mid n,m \ge 1\,\},$$ as well as the languages
$\{\,a^nw \mid n \ge 1 \wedge w \in (b^*c^*)^kb^* \wedge |w|_b=n\,\}$,
for all constants $k\geq 0$. All of these languages are either
semilinear or non-bounded. But in contrast to many other 
computational devices, for example certain multi-head finite
automata, parallel communicating finite automata, and certain parallel communicating
grammar systems, $\mcoca{1}$s can accept non-semilinear bounded languages.

\begin{example}\label{exmpl:sqrt}
The language $L_1=\{\,a^n b^{n+\lfloor \sqrt{n} \rfloor} \mid n \ge 1\,\}$
belongs to the family $\rtf(\mcoca{1})$.

In~\cite{Mazoyer:1994:SODCA} a $\ca$ is constructed such that 
its cell $n$ enters a designated state exactly at time step
\hbox{$2n\!+\!\lfloor \!\sqrt{n} \rfloor$}, and at most $n$ cells are used for the computation. 
In fact, the $\ca$ constructed is actually an $\oca$.
Additionally, each cell performs only a finite
number of communication steps. Thus, the $\ca$ constructed 
is an $\mcoca{1}$.

An $\mcoca{1}$ accepting $L_1$ implements the above construction 
on the \mbox{$a$-cells} of the input $a^nb^m$. 
Thus, the leftmost cell enters the designated state~$q$ at time step 
$2n+\lfloor \sqrt{n} \rfloor$. Additionally, in the rightmost cell a signal $s$
with maximum speed is sent to the left. When this signal arrives in an $a$-cell
exactly at a time step at which the cell would enter the designated state $q$,
the cell changes to an accepting state instead. 
So, if $m=n+\lfloor \sqrt{n} \rfloor$, then $s$ arrives at time 
$2n+\lfloor \sqrt{n} \rfloor$ at the leftmost cell and the input is accepted.
In all other cases the input is rejected. Clearly, the $\oca$ constructed is 
an $\mcoca{1}$.
\hfill$\diamond$
\end{example}

\begin{example}\label{exmpl:2to:the:n}
The language $L_2=\{\,a^{2^n} b^{n} c^{2^n+n} \mid n \ge 1\,\}$
belongs to the family $\rtf(\mcoca{1})$.

The rough idea of the construction is sketched as follows. We first describe
a real-time $\mcca{1}$ accepting $\{\,b^n a^{2^n} \mid n \ge 1\,\}$. Then this
two-way real-time $\mcca{1}$ is simulated by a one-way linear-time $\mcoca{1}$
accepting the reversal language $\{\,a^{2^n} b^n \mid n \ge 1\,\}$
in time $2 \cdot (n+2^n)$. The time additionally needed is provided by
adding the suffix $c^{2^n+n}$ to the input. Finally, the correct length
of the suffix is checked.

In more detail, we first consider the construction of a signal with speed $2^n$
given in~\cite{Mazoyer:1994:SODCA}. There a $\ca$ is described whose cell $n$ 
enters a designated state exactly at time step $2^n$, where at most $2^n$ 
cells are used for the computation. The $\ca$ constructed is in fact
an $\mcca{1}$. Similar to the construction in Example~\ref{exmpl:sqrt} 
we implement the construction of the signal $2^n$ on the $b$-cells, and 
the rightmost \mbox{$a$-cell} sends a signal $s$ with maximum speed
to the left. We know that the rightmost $b$-cell enters a designated state~$q$ exactly
at time step $2^n$ if the input is $b^na^m$. Moreover, signal $s$
arrives at time step $m+1$ in cell~$n$. If $m=2^n$, then cell $n$ has entered
the state~$q$ exactly one time step before, and now changes to an accepting
state which
is sent with maximum speed to the leftmost cell. In all other cases the input
is rejected. Altogether, we derive that $\{\,b^n a^{2^n} \mid n \ge 1\,\} \in 
\rtf(\mcca{1})$.

Next, we want to accept $L_2$ by some real-time $\mcoca{1}$. To this end, 
we utilize the fact that the reversal of every language accepted by some 
real-time $\ca$ can be accepted in twice the time by some 
$\oca$~\cite{Choffrut:1984:RTCATA,kutrib:2008:ca-cpv,Umeo:1982:DOWSTWRTCARP}. 
The essence of the proof is that each cell of the one-way device
collects the states of its both neighbors to the right in an additional time
step. With this information it can simulate the behavior of its immediate 
neighbor to the right in the two-way device.
In this way, $n$ steps of the $\ca$ can be simulated in $2n$ steps by the 
$\oca$ on reversed input. It follows that the resulting $\oca$ is an
$\mcoca{1}$ if the given $\ca$ was an $\mcca{1}$. In particular, we obtain 
that $\{\,a^{2^n} b^n \mid n \ge 1\,\}$ can be accepted by some $\mcoca{1}$ 
in time $2 \cdot (n+2^n)$. Thus, the simulation can be performed
on input $a^{2^n} b^{n} c^{2^n+n}$ in real-time.

Finally, the number of $c$s remains to be checked. To this end,
we consider the already mentioned language $\{\,a^nb^n \mid n \ge 1\,\}$ 
which belongs to the family $\rtf(\mcoca{1})$~\cite{Kutrib:2009:cssc:proc}.
Here we match the number of $a$s and $b$s against the number of $c$s on an 
additional track. So, we obtain the desired $\mcoca{1}$.
\hfill$\diamond$
\end{example}

%

\section{Decidability questions}\label{sect:decidability}

This section is devoted to decidability problems. In fact, the results
show undecidability of various questions  
for real-time $\mcoca{\log n}$s and $\scoca{n}$s accepting bounded languages. 
First we show that emptiness is undecidable for real-time $\mcoca{\log n}$s
and $\scoca{n}$s  
accepting bounded languages by reduction 
from Hilbert's tenth problem which is known to be undecidable. The problem is
to decide whether a given polynomial
$p(x_1,\ldots,x_n)$ with integer coefficients has an integral root. That is, to
decide whether there are integers $\alpha_1,
\ldots, \alpha_n$ such that $p(\alpha_1,\ldots,\alpha_n)=0$. 
In~\cite{ibarra:1978:rbmcmdp} Hilbert's tenth problem was used to show
that emptiness is undecidable for certain multi-counter machines.
As is remarked in \cite{ibarra:1978:rbmcmdp}, it is sufficient to restrict the variables 
$x_1, \ldots, x_n$ to take non-negative integers only. If $p(x_1,\ldots,x_n)$ 
contains a constant summand, then we may assume that it has a negative sign.
Otherwise, $p(x_1,\ldots,x_n)$ is multiplied by $-1$.
Then, such a polynomial has the following form:
$
p(x_1,\ldots,x_n)=t_1(x_1,\ldots,x_n)+\cdots+t_r(x_1,\ldots,x_n)
$,
where each $t_j(x_1,\ldots,x_n)$ ($1 \le j \le r$) is a term of the form
$
s_jx_1^{i_{j,1}} \ldots x_n^{i_{j,n}}
$
with $s_j \in \{+1,-1\}$ and $i_{j_1},\ldots,i_{j_n} \ge 0$.
It should be remarked that some terms~$t_j$ may be equal.
Additionally, we may assume that the summands are ordered according to their sign, 
i.\,e., there exists $1 \le p \le r$ such that $s_1=\cdots=s_p=1$ and 
$s_{p+1}=\cdots=s_r=-1$. Moreover, constant terms occur only at the
end of the sum. I.\,e., $t_r=\cdots=t_{r-c+1}=-1$, if $p$ contains $c>0$ constant
terms. 

We first look at the positive terms~$t_j$, $1 \le j \le p$, and define 
languages $L(t_j)$ as follows.
$$
L(t_j)=\{\,b_1^{\alpha_1} \cdots b_n^{\alpha_n} 
c_1^{\alpha_1^{i_{j,1}}} \cdots c_n^{\alpha_n^{i_{j,n}}} 
d_1^{2^n \cdot \alpha_1^{i_{j,1}} \cdots \alpha_n^{i_{j,n}}}\cent
\mid \alpha_1,\ldots,\alpha_n \ge 0\,\}
$$
For the negative, non-constant terms $t_j$ with $p+1 \le j \le r$ the
definition of~$L(t_j)$ is identical except for the fact that each symbol $d_1$ 
is replaced by some symbol~$d_2$. For each negative, constant term $t_j$, 
we define $L(t_j)=\{d_2^{2^n}\cent\}$.
Since $n$ is a constant depending on the given polynomial $p$, we can observe 
that each $L(t_j)$ is a bounded language.

\begin{lemma}\label{lemma:ltj} 
For $1 \le j \le r$, the language $L(t_j)$ belongs to the family $\rtf(\ca)$.
\end{lemma}

\begin{proof}
If $t_j$ is a constant term, then $L(t_j)$ is a regular language and belongs to 
$\rtf(\ca)$.
Otherwise, $L(t_j)$ can be accepted by a real-time CA as follows. 
We may assume that the input is correctly formatted, since this can be checked 
with some leftward signal starting in the rightmost cell.
 
The first task is to check for $1 \le m \le n$ that the number of symbols 
$c_m$ is equal to the number of symbols $b_m$ to the power $i_{j,m}$, that is,
the number of symbols $c_m$ is equal to $\alpha_m^{i_{j,m}}$. 
In~\cite{Mazoyer:1994:SODCA} for every $k \geq 2$, a two-way $\ca$ is 
constructed whose leftmost cell enters a designated state $q_m$ exactly 
at every time step $x^k$, for $x\geq 1$. The $\ca$s do not use more than $x$ 
cells. 

Since $i_{j,m}$ is a constant, we can implement this construction for
$k={i_{j,m}}$ on the $c_m$-cells, for all $1 \le m \le n$.
Whenever the leftmost cell of the block consisting of $c_m$-cells, 
called $c_m$-block, enters the designated state $q_m$, a signal is 
sent to the left which marks one symbol $b_m$. Additionally, the 
rightmost cell of the $c_m$-block sends a signal $s_m$ with maximum 
speed to the left. When $s_m$ arrives in the leftmost cell of the $c_m$-block,
it checks whether the cell would enter the state $q_m$ at that time step. 
Moreover, the signal $s_m$ checks that all symbols $b_m$ have been marked
and no signal failed to mark a $b_m$ since there are too few of them. The
latter can be observed and remembered by the cell carrying the leftmost
$b_m$.

The second task is to check that the number of symbols $d_1$ is equal
to \mbox{$2^n \cdot \alpha_1^{i_{j,1}} \cdots \alpha_n^{i_{j,n}}$.} 
We next describe the construction and remark that the construction
for symbols $d_2$ is identical. 
The principal idea is as follows: 
In the rightmost cell of a $c_m$-block a signal $s'_m$ is set up that 
moves through the block back and forth. The signal $s'_n$ moves with 
maximum speed. The signal $s'_m$ through some other
$c_m$-block is stepped by the arrival of the signal $s'_{m+1}$ at 
its leftmost cell. Whenever this happens,
an auxiliary signal is sent from the rightmost $c_m$-cell to the left
that moves the signal $s'_m$ one cell.
Clearly, by dropping $s'_n$ the whole 
process gets frozen. 

Now, the rightmost cell of the array emits a signal $s_d$ to the left at initial
time. When $s_d$ arrives at the rightmost $c_n$-cell at the same time
as signal $s'_n$, the latter is dropped. Now, signal $s_d$ continues
to move to the left and checks that all signals~$s'_m$ stay 
at the rightmost $c_m$-cell and, 
additionally, that the $c_1$-block has been passed through 
back and forth by $s'_1$ exactly once, which can be remembered by $s'_1$. 
In this case, the number of $d_1$s is exactly
$2\cdot \alpha_1^{i_{j,1}} \cdot 2\cdot \alpha_2^{i_{j,2}}\cdots
2\cdot \alpha_2^{i_{j,n}} = 2^n \cdot \alpha_1^{i_{j,1}} \cdots \alpha_n^{i_{j,n}}$.
A schematic computation may be found in Figure~\ref{fig:product}.

Now, we can construct a real-time CA accepting $L(t_j)$ by realizing both tasks on 
different tracks and checking the correctness of the input format.
\end{proof}

\begin{figure}[ht!]
\centering
\includegraphics[scale=.7]{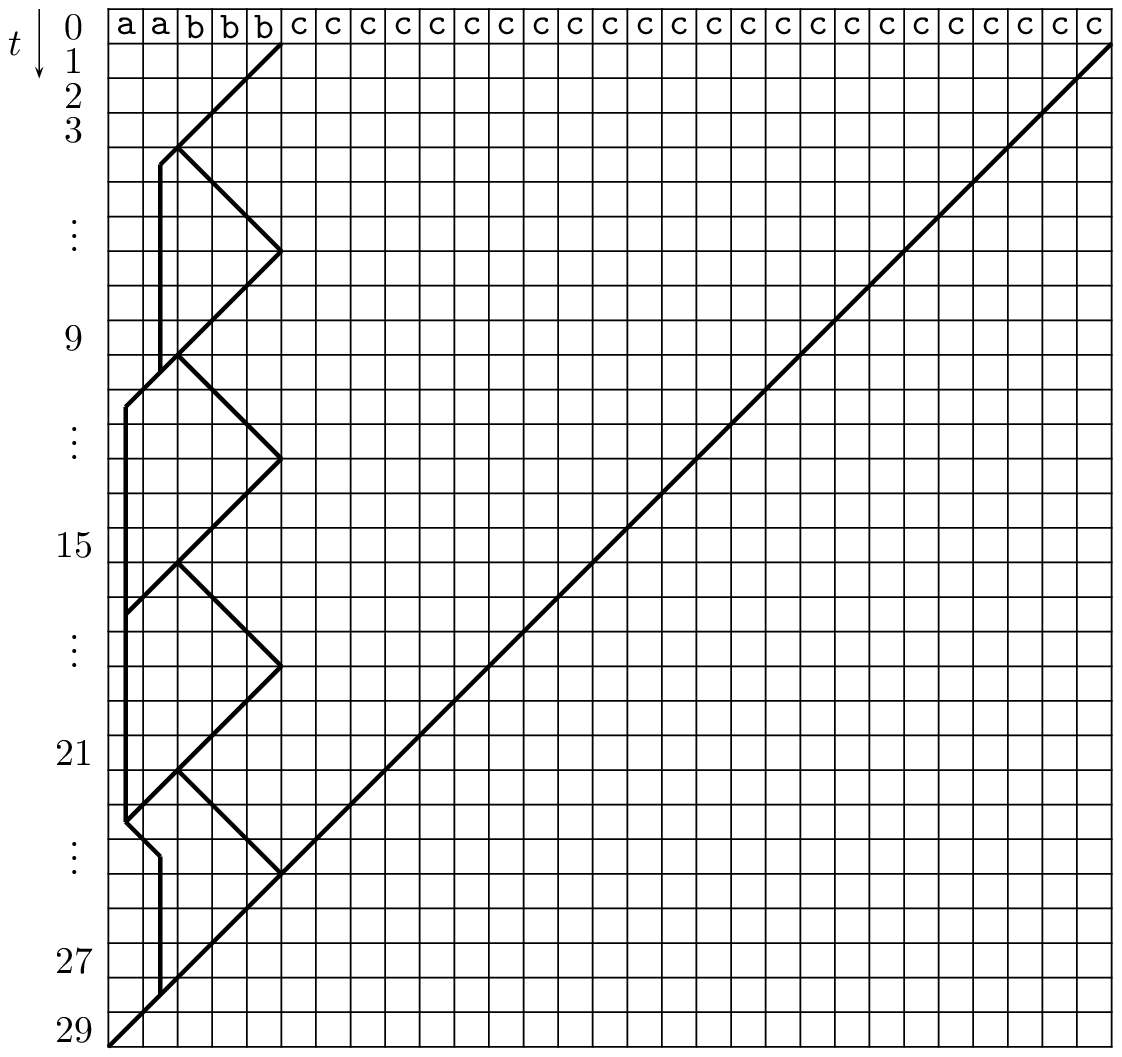}
\caption{Schematic computation of $a^2 b^3 c^{2^2 \cdot 2 \cdot 3}=a^2 b^3 c^{24}$.}
\label{fig:product}
\end{figure}

We next consider the following regular languages $R_j$ depending on the sign of 
$t_j$.
If $s_j=1$, then $R_j=b_{1}^*\cdots b_{n}^* c_{1}^*\cdots c_{n}^* d_1^* \cent$. 
If $s_j=-1$ and $t_j$ is non-constant, then 
$R_j=b_{1}^*\cdots b_{n}^* c_{1}^*\cdots c_{n}^* d_2^* \cent$. 
Otherwise, we set $R_j =d_2^* \cent$.
Then, we define for positive terms $t_j$
\begin{multline*}
\tilde{L}(t_j)=\{\,
a_1^{\,\alpha_1} \cdots a_n^{\alpha_n} w_1 \cdots w_{j-1}
b_1^{\alpha_1} \cdots b_n^{\alpha_n} c_1^{\alpha_1^{i_{j_1}}} 
\cdots c_n^{\alpha_n^{i_{j_n}}} 
d_1^{2^n \cdot \alpha_1^{i_{j_1}} \cdots \alpha_n^{i_{j_n}}}\cent 
w_{j+1} \cdots w_r \ \mid \alpha_1,\dots,\alpha_n 
\ge 0\\ \textrm{ and } w_i \in R_i \textrm{ for } 1 \le i \le r \textrm{ with } i \neq j
\,\}.
\end{multline*}
The languages $\tilde{L}(t_j)$ for negative, non-constant terms are defined 
analogously.  For negative, constant terms we define
\begin{multline*}
\tilde{L}(t_j)=\{\,
a_1^{\alpha_1} \cdots a_n^{\alpha_n} w_1 \cdots w_{j-1}
d_2^{2^n}\cent w_{j+1} \cdots w_r \ \mid \alpha_1,\dots,\alpha_n 
\ge 0\\ \textrm{ and } w_i \in R_i \textrm{ for } 1 \le i \le r \textrm{ with } i \neq j
\,\}.
\end{multline*}
Now, we consider the language $\tilde{L}(p)=\bigcap_{i=1}^r \tilde{L}(t_j)$
and observe that $\tilde{L}(t_j)$ and, thus, $\tilde{L}(p)$ are still bounded 
languages (up to a renaming of symbols).

\begin{lemma}\label{lemma:ltildep} 
The language $\tilde{L}(p)$ belongs to the family $\rtf(\ca)$.
\end{lemma}

\begin{proof} 
Since $\rtf(\ca)$ is closed under intersection, we have to show that each
language $\tilde{L}(t_j)$ belongs to $\rtf(\ca)$. First, we can observe that
$\tilde{L}(t_j)$ is a regular language and, thus, is accepted by a real-time CA, if
$t_j$ is a negative, constant term. We next present the construction for $t_j$
being a non-constant term. 

Due to Lemma~\ref{lemma:ltj} we can construct a real-time CA accepting $L(t_j)$.
Then we concatenate the regular language $a_1^* \cdots a_n^* R_1 \cdots R_{j-1}$
on the left and the regular language $R_{j+1} \cdots R_{r}$ on the right
of $L(t_j)$. Since the family $\rtf(\ca)$ is closed under 
marked concatenation~\cite{Seidel:1979:LRSCA}, we again obtain
a real-time $\ca$ language.
It remains to be shown that the number of $a_i$s is equal to the number of 
$b_i$s, for $1 \le i \le n$. 
It is known that the language $\{\,a^nb^n \mid n \ge 1\,\}$ can be accepted by a 
real-time CA. The principal idea is to send the $b$s to the left and to match 
them against the $a$s. If both numbers are equal, the input is accepted and 
otherwise rejected. By an obvious generalization of this idea, we can check 
with a real-time CA that the number of $a_i$s is equal to the number of 
$b_i$s.
Additionally, we have to take care that $b_i$s occurring in $w_1, \ldots, w_{j-1}$
are not matched against $a_i$s. To this end, the cells carrying $a_i$s are
ignoring the first $j-1$ blocks of $b_i$s. Altogether, language $\tilde{L}(t_j)$ can
be accepted by a real-time CA.
\end{proof}

Finally, let $L(p)=\{\,w \in \tilde{L}(p) \mid |w|_{d_1}=|w|_{d_2}\,\}$.

\begin{lemma}\label{lemma:lp}
The language $L(p)$ belongs to the family $\rtf(\ca)$.
\end{lemma}

\begin{proof}
Here we have to check that the input belongs to $\tilde{L}(p)$ and that the number
of occurring symbols $d_1$ is equal to the number of occurring symbols $d_2$.
The former task can be realized by some real-time CA due to Lemma~\ref{lemma:ltildep}.
The latter task can also be realized by some real-time CA: Due to the format of words in
$\tilde{L}(p)$ we know that each word can be divided into two parts. The first
part contains symbols $d_1$ and no symbols $d_2$ whereas the second part contains
symbols $d_2$ and no symbols $d_1$. By sending symbols $d_2$ to the left and
matching them against symbols $d_1$, we can check that their number is equal.
\end{proof}

\begin{lemma}\label{lemma:lp1}
The language $L_1(p)=\{\,w a^{|w|}b^{2^{|w|}} \mid w \in L(p)^R\,\}$ 
belongs to the family $\rtf(\scoca{n})$.
\end{lemma}

\begin{proof}
Let us first give evidence that the language $\{\,a^mb^{2^m} \mid m \ge 1\,\}$ 
belongs to $\rtf(\scoca{n})$.
To this end, we implement a binary counter in the $a$-cells, i.e, we store the 
binary encoding of the currently counted value in these cells. The least significant bit 
is simulated in the rightmost $a$-cell. The information to be communicated to
the left are carry-overs. Now, the counter is increased at every time step.
Furthermore, the rightmost $b$-cell emits a signal to the left at initial time.
When this signal arrives at the rightmost $a$-cell, it checks successively 
whether all $a$-cells passed through are in a state indicating that they have produced 
a carry-over before. If it arrives in an $a$-cell that is in a carry-over
state for the first time, that cell enters an accepting state.
An example computation is depicted in Figure~\ref{fig:bincount}. 

\begin{figure}[ht]
\centering
\includegraphics[scale=.7]{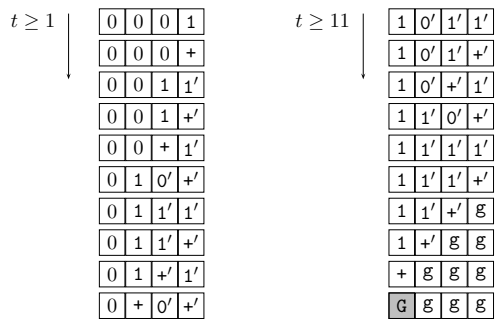}
\caption{A binary counter accepting $b^{2^4}=b^{16}$.
A $\texttt{+}$ denotes a carry-over and some primed state indicates that the cell has 
produced a carry-over at some time before. The latter is checked by signal
$g$, where $G$ indicates an accepting state.}
\label{fig:bincount}
\end{figure}

Let us consider an input $a^mb^{2^m}$. The rightmost $a$-cell performs $2^{m-1}$
communication steps to send the carry-overs during the counting phase. 
Its left neighbor performs $2^{m-2}$ communication steps
and so on. Thus, we obtain that all $a$-cells perform 
not more than $\sum_{i=1}^m 2^{i-1} \in O(2^m)$ communication steps.
Furthermore, the $b$-cells perform only a constant number of communication steps.
Altogether, $O(2^m)$ communication steps are performed in total,
and an $\scoca{n}$ can be constructed.

Next, we consider the language $\{\,w a^{|w|} \mid w \in L(p)^R\,\}$. The correct
number of $a$s can be checked in a similar way as for the language 
$\{\,a^mb^m \mid m \ge 1\,\}$. This check costs a constant
number of communication steps per cell.
In the proof of Lemma~\ref{lemma:lp} a 
real-time CA is constructed which accepts $L(p)$. 
It is known that an OCA accepting $L(p)^R$ in twice the time can be 
constructed~\cite{Choffrut:1984:RTCATA,kutrib:2008:ca-cpv,Umeo:1982:DOWSTWRTCARP}.
We implement this construction in order to check the $w$. By concatenating
the additional symbols $a$, the $\oca$ works still in real time.
The rightmost $a$-cell additionally sends a signal which
freezes the computation in the first $|w|$ cells.

Now, it is easy to construct a real-time OCA accepting $L_1(p)$. To conclude the
proof we have to show that the real-time OCA constructed is a real-time $\scoca{n}$.
The $b$-cells perform a constant number of communication steps per cell. The
$a$-cells perform $O(2^m)$ communication steps to realize the counter. Additionally,
a constant number of communications per cell is needed to check the length of $w$
with the number of $a$s and to send the freezing signal. Finally, each of the first $|w|$
cells can perform at most~$2m$ communication steps due to the freezing
signal. In total, they perform at most $2m^2$ communication steps. Altogether,
we obtain that at most $O(2^m)$ communication steps are performed which shows
that $L_1(p) \in \rtf(\scoca{n})$.
\end{proof}

\begin{lemma}\label{lemma:lp2}
The language $L_2(p)=\{\,w a^{|w|}b^{2^{|w|}}c^{|w|} \mid w \in L(p)^R \,\}$ 
belongs to the family $$\rtf(\mcoca{\log n}).$$
\end{lemma}

\begin{proof}
Let us first consider the language $\{\,b^{2^m}c^m \mid m \ge 1\,\}$.
It is shown in~\cite{kutrib:2008:ca-cpv} that the unary language 
$\{\,b^{2^m} \mid m \ge 1\,\}$ can be accepted by a $(2^m+m)$-time OCA. 
The principal idea is to construct a moving binary counter which starts 
with one time step delay in the rightmost cell and moves to the left, 
whereby the cells passed through are counted. 
If necessary, the length of the counter is increased. 
Additionally, the cells passed through check whether all bits are 1.
In this case, some number $2^m-1$ has been counted. 
Taking into account the delayed start we obtain that some number $2^m$ 
has been counted. In this case a cell enters an accepting state. An example
computation is depicted in Figure~\ref{fig:rtlog}. 
Since the counter moves from right to left
and all cells passed through enter an accepting or rejecting permanent
state, we can observe that the number of communication steps of each cell 
is bounded by the length of the counter, that is, by the logarithm of the 
length of the input.

\begin{figure}[ht]
\centering
\includegraphics[scale=.6]{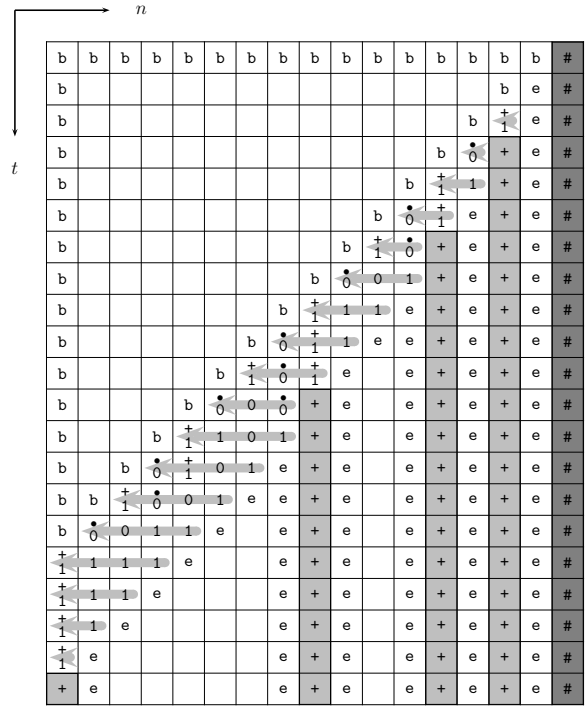}
\caption{Computation of an OCA accepting an input of the 
language \mbox{$\{\,b^{2^m} \mid m \ge 1\,\}$}
in time $2^m+ m$.}
\label{fig:rtlog}
\end{figure}

So, the language $\{\,b^{2^m} \mid m \ge 1\,\}$ is accepted by an 
$\mcoca{\log n}$ which needs more than
real time. It is easy to modify the construction such that 
the language $\{\,b^{2^m}c^m \mid m \ge 1\,\}$
is accepted in real time. 

Next, let us consider the language $\{\,w a^{|w|} \mid w \in L(p)^R \,\}$. 
In the proof of Lemma~\ref{lemma:lp1} an $\scoca{n}$ accepting this 
language is constructed.
So, we obtain that $\{\,w a^{|w|} \mid w \in L(p)^R \,\}$
belongs to $\rtf(\oca)$.

Finally, we concatenate both languages and check the same number of $a$s and~$c$s by
sending $c$s to the left and matching them against $a$s. Altogether, we obtain a
real-time OCA accepting $L_2(p)$. To conclude the proof we have to show that the 
real-time OCA constructed is in fact a real-time $\mcoca{\log n}$.
The length of the input is $3m+2^m$. Thus, we have to show that each cell does not
perform more than $O(m)$ communication steps.
Concerning the $b$-cells not more than $O(m)$ communication steps are performed due
to the construction provided. The matching of $c$s against $a$s 
causes not more than $O(m)$ additional communication steps. Therefore, the number of
communication steps in $b$- and $c$-cells is of order $O(m)$.
The $a$-cells receive $m$ signals from the $c$-cells and send a block of $a$s to
be matched against $w$. Altogether, not more than $O(m)$ communication steps are performed.
Finally, due to the freezing signal the first $|w|$ cells can perform at most 
$2m \in O(m)$ communication steps
after the computation whether $w \in L(p)^R$.
Altogether, we obtain that $L_2(p)$ belongs to $\rtf(\mcoca{\log n})$.%
\end{proof}

Now, we are prepared to derive the undecidability results.

\begin{theorem}
Given an arbitrary real-time $\scoca{n}$ or $\mcoca{\log n}$ $\mathcal{M}$
accepting a bounded language, it is undecidable whether
$L(\mathcal{M})$ is empty.
\end{theorem}

\begin{proof}
Due to Lemma~\ref{lemma:lp1} we can construct a real-time $\scoca{n}$ $\mathcal{M}$ 
accepting $L_1(p)$. 
By the construction of $L_1(p)$, it is not difficult to observe that $\mathcal{M}$ 
accepts the empty set if and only if $2^n \cdot p(x_1,\ldots,x_n)$
has no solution in the non-negative integers. The latter is true if and
only if $p(x_1,\ldots,x_n)$ has no solution in the non-negative integers.
Since Hilbert's tenth problem is 
undecidable, we obtain
that the emptiness problem for real-time $\scoca{n}$s is undecidable.
The argumentation for $\mcoca{\log n}$ is similar considering $L_2(p)$ and 
Lemma~\ref{lemma:lp2}.
\end{proof}

By standard techniques (see, e.\,g.,~\cite{kutrib:2009:calt}) one can show
the following results.

\begin{theorem}\label{theo:undec}
The problems of testing finiteness, infiniteness, universality,
inclusion, equivalence, regularity, and context-freedom are undecidable for  
arbitrary real-time $\scoca{n}$s and $\mcoca{\log n}$s accepting bounded languages.
\end{theorem}

\bibliographystyle{eptcs}
\bibliography{kutrib}

\end{document}